\documentclass[11pt,onecolumn]{article}
\usepackage{times}
\usepackage{amsmath}
\usepackage{amssymb}
\usepackage{xspace}
\usepackage{theorem}
\usepackage{graphicx}
\usepackage{ifpdf}
\usepackage{url,hyperref}
\usepackage{latexsym}
\usepackage{euscript}
\usepackage{xspace}
\usepackage{color}
\usepackage{makeidx}
\usepackage{picins,wrapfig}
\usepackage{stackrel}
\usepackage{algorithmic, algorithm}

\long\def\remove#1{}

\newtheorem{theorem}{Theorem}[section] 

\newtheorem{proposition}[theorem]{Proposition}
\newenvironment{proof}{{\em Proof:}}{\hfill{\hfill\rule{2mm}{2mm}}}

\newcommand {\mm}[1] {\ifmmode{#1}\else{\mbox{\(#1\)}}\fi}

\newcommand{\eps}{{\varepsilon}}

\newcommand{\TT}		{\mathrm {\mathbb{T}}}

\newcommand{\g}			{g}

\newcommand{\KK}		{{\cal K}}

\newcommand{\rank}		{\mm {\rm rank}}

\newcommand{\Tub}		{\mathrm {Tub}}
\newcommand{\length}		{\mm {\rm Len}}

\newcommand{\SP}	{\Pi}
\renewcommand{\Vert}	{\mathrm{Vert}}

\newcommand{\concatenate}		{\circ}
\newcommand{\Cech}  {{\cal C}}

\newcommand{\Rips}	{{\cal R}}
\newcommand{\homo}	{{\sf H}_1}
\newcommand{\homok}	{{\sf H}_k}

\newcommand{\mainAlg}		{{\sc {ShortLoop}}\xspace}
\newcommand{\CanonGen}		{{\sc {CanonGen}}\xspace}
\newcommand{\SPGen}		{{\sc {SPGen}}\xspace}

\newcommand{\hsim}		{\approx}

\newcommand {\myparagraph}[1] {\vspace{0.15in} \noindent {\bf {#1}}}

\begin{document}

\title{Approximating Loops in a Shortest Homology Basis from Point Data}

\author{
Tamal K. Dey\thanks{
Department of Computer Science and Engineering,
The Ohio State University, Columbus, OH 43210, USA.
Email: {\tt tamaldey@cse.ohio-state.edu}}
\quad\quad
Jian Sun\thanks{
Department of Computer Science,
Stanford University, Palo Alto, CA 94305, USA.
Current address: Department of Computer Science, Princeton University, Princeton,
NJ 08544, Email: {\tt jiansun@cs.princeton.edu}}
\quad\quad Yusu Wang\thanks{
Department of Computer Science and Engineering,
The Ohio State University, Columbus, OH 43210, USA.
Email: {\tt yusu@cse.ohio-state.edu}}
}

\date{}
\maketitle

\begin{abstract}
Inference of topological and geometric attributes of a
hidden manifold from its point data is a fundamental problem
arising in many scientific studies and engineering applications.
In this paper we present an algorithm to compute a set of
loops from a point data that presumably
sample a smooth manifold $M\subset \mathbb{R}^d$. 
These loops approximate a {\em shortest} basis
of the one dimensional homology group $\homo(M)$ 
over coefficients in finite field $\mathbb{Z}_2$.
Previous results addressed the issue of computing the rank
of the homology groups from point data, but there is no result
on approximating the shortest basis of a manifold from its
point sample. In arriving our result, we also present a 
polynomial time algorithm for computing a shortest basis
of $\homo(\KK)$ for any finite {\em simplicial complex} $\KK$ 
whose edges have non-negative weights.
\end{abstract}
\thispagestyle{empty}
\setcounter{page}{0}
\newpage

\section{Introduction}
Inference of unknown structures  
from point data is a fundamental problem in many areas of
science and engineering that has motivated 
wide spread research~\cite{AB98,CEH07,pca,NSW06,mds,ZC05}. 
Typically, this data is assumed to be sampled from
a manifold sitting in a high dimensional space whose
geometric and topological properties are to be derived
from the data. In this work, we are particularly
interested in computing a set of loops from data which not only
captures the topology but is also aware of the geometry
of the sampled manifold. Specifically, we aim to approximate 
a shortest basis for the one dimensional
homology group from the data. 

Recently, a few algorithms for computing homology groups from point data
have been developed.
One approach would be to reconstruct the sampled space from its 
point data~\cite{BGO07,CL06,CDR05}
and then apply known techniques for homology computations
on triangulations~\cite{Hatcher}. However, this
option is not very attractive since a full-blown reconstruction
with known techniques requires costly computations
with Delaunay triangulations in high dimensions.
Chazal and Oudot~\cite{CO08} showed how one can use  
less constrained data structures such as Rips, \v{C}ech, and witness
complexes to infer the rank of the homology groups
by leveraging persistence algorithms~\cite{ELZ02,ZC05}.
Among these, the Rips complexes are the easiest to compute though
they consume more space than the others, an issue
which has started to be addressed~\cite{DL09}.

All of the above works so far considered only focus on computing the
Betti numbers, the rank of the homology groups. Although the
persistence algorithms~\cite{ELZ02,ZC05} also provide representative
cycles of a homology basis, they remain oblivious to the
geometry of the manifold. As a result, these cycles do not
have nice geometric properties. 
A natural question to pose is that if the loops of 
the one dimensional homology group
are associated with a length under some metric, can one
approximate/compute
a shortest set of loops that generate the homology group
in polynomial time?
This question has been answered in affirmative for the special
case of surfaces when they are represented with triangulations~\cite{EW05}.
In fact, considerable progress has been made for this special
case on various versions of the problem.
We cannot apply these techniques, mainly because we deal with point data
instead of an input triangulation. Also, these
works either consider a surface~\cite{CEN092,CEN09,CE06,EW05} 
instead of a manifold
of arbitrary dimension in an Euclidean space, or use
a local measure other than the lengths of the
generators in a basis~\cite{CF09}. 

Our main result is an algorithm that can compute a set of
loops from a Rips complex of the given data and a proof
that the lengths of the computed loops
approximate those of a shortest basis of the one dimensional
homology group of the sampled manifold. In arriving at this
result, we also show how to compute a shortest basis for
the one dimensional homology group of any finite {\em simplicial complex}
whose edges have non-negative weights. Given that computing
a shortest basis for $k$-dimensional homology groups of a simplicial complex 
over $\mathbb{Z}_2$ coefficients is NP-hard for $k\geq 2$ 
(Chen and Freedman~\cite{CF10}), this result settles the open case for $k=1$. 

\subsection{Background and notations}
We use the concepts of homology groups, \v{C}ech and Rips
complexes from algebraic topology and geodesics from differential
geometry. We briefly discuss them and introduce relevant notations 
here; the readers can obtain the details from any standard
book on the topics such as~\cite{docarmo,Hatcher}.\\ 

\noindent
{\bf Homology groups and generators}: A homology group of a topological
space $\TT$ encodes its topological connectivity. We use 
$\homok(\TT)$ to denote its $k$-dimensional homology group
over the coefficients in $\mathbb{Z}_2$. Since $\mathbb{Z}_2$
is a field, $\homok(\TT)$ is a vector space of dimension $k$ and
hence admits a basis of size $k$. We are concerned with the
$1$-dimensional homology groups $\homo(\TT)$. The elements of
$\homo(\TT)$ are equivalent classes $[g]$ of $1$-dimensional
cycles $g$, also called {\em loops}. A set $\{[g_1],\ldots,[g_k]\}$
generating $\homo(\TT)$ is called its basis where $k=rank(\homo(\TT))$.
Simplifying the notation, we say $\{g_1,\ldots,g_a\}$ generate
$\homo(\TT)$ if $\{[g_1],\ldots,[g_a]\}$ generate $\homo(\TT)$
and is a basis if $a=rank(\homo(\TT))$. We assume that each loop $g$ in
$\TT$ is associated with a non-negative weight $w(g)$. The length of a set of
loops $G=\{g_1,\ldots,g_a\}$ is given by
$\length(G)=\Sigma_{i=1}^a w(g_i)$. A {\em shortest set of
generators} or a {\em shortest basis} of $\homo(\TT)$ 
is a basis $G$ of $\homo(\TT)$ where 
$\length(G)$ is minimal over all bases. 
When $\TT$ is a simplicial complex,
all loops are restricted to its $1$-skeleton.\\
 
\noindent
{\bf Complexes}:
Let $B(p,r)$ denote an open Euclidean $d$-ball centered
at $p$ with radius $r$. For a point set $P\subset \mathbb{R}^d$,
and a real $r>0$, the \v{C}ech complex
$\Cech^{r}(P)$ is a simplicial complex where a simplex
$\sigma\in \Cech^{r}(P)$ if and only if  
${\rm Vert(\sigma)}$, the vertices of $\sigma$,
are in $P$ and are the centers of $d$-balls of radius
$r/2$ which have a non-empty common intersection,
that is, $\cap_{p\in {\rm Vert}(\sigma)} B(p,r/2)\not=\emptyset$.
Instead of common intersection, if we only require
pairwise intersection among the $d$-balls, we get
the Rips complex $\Rips^{r}(P)$.
It is well known that the two complexes are related by a 
nesting property:

\begin{proposition}
For any finite set $P\subset \mathbb{R}^d$ and any $r\geq 0$, one has
$\Cech^{r}(P) \subseteq \Rips^{r}(P) \subseteq
\Cech^{2r}(P)$.
\end{proposition}

\noindent
{\bf Geodesics}:
The vertex set $P$ of the simplicial complexes we consider is
a dense sample of a smooth compact manifold $M\subset\mathbb{R}^d$
without boundary. 
Assume that $M$ is isometrically embedded, that is, $M$ inherits the metric
from $\mathbb{R}^d$. For two points $p,q\in M$, a 
{\em geodesic} is a curve connecting $p$ and $q$ in $M$ whose
acceleration has no component in the tangent spaces
of $M$. Two points may have more than one geodesic among which
the ones with the minimum length are called {\em minimizing geodesics}. 
Since $M$ is compact, any two points admit a minimizing geodesic.
If $p$ and $q$ are close enough, this minimizing geodesic is unique,
which we denote as $\gamma(p,q)$. The lengths of minimizing
geodesics induce a distance metric $d_M:M\times M\rightarrow \mathbb{R}$
where $d_M(p,q)$ is the length of a minimizing geodesic between
$p$ and $q$. Clearly, $d(p,q)\leq d_M(p,q)$ where $d(p,q)$ is the
Euclidean distance. If $d(p,q)$ is small, Proposition~\ref{euclid-geod}
asserts that there is an upper bound on $d_M(p,q)$ in terms of $d(p,q)$. 
Our proof extends a result
in~\cite{BSW08} where Belkin et al. show 
the same result on a surface in $ \mathbb{R}^3$. 
See the Appendix for the proof. 
The {\em reach} $\rho(M)$ is defined as the minimum distance
between $M$ and its medial axis.

\begin{proposition}
If $d(p,q)\leq \rho(M)/2$, one has $d_M(p,q)\leq (1+\frac{4d^2(p, q)}{3\rho^2(M)})d(p, q)$.
\label{euclid-geod}
\end{proposition}

\noindent
{\bf Convexity radius and sampling}:
For a point $p\in M$,
the set of all points $q$ with $d_M(p,q)< r$ 
form $p$'s {\em geodesic ball} $B_M(p,r)$
of radius $r$.
It is known that
there is a positive real $r_p$ for each point
$p\in M$ so that $B_M(p,r_p)$ is {\em convex} in the sense
that the minimizing geodesics between any two points
in $B_M(p,r_p)$ lie in $B_M(p,r_p)$. 
The {\em convexity radius} of $M$ is $\rho_c(M)=\inf_{p\in M} r_p$. 
We use Euclidean distances to define the sampling density.
We say a discrete set $P\subset M$ is an $\eps$-sample\footnote{
Here $\eps$-sample is not defined relative to reach
or feature size as commonly done in reconstruction literature
~\cite{AB98,CL06,CDR05}.} of
$M$ if $B(x,\eps)\cap P\not=\emptyset$ for each point $x\in M$.
 
\subsection{Main results}
We present an algorithm that
computes a set of loops
$G=\{g_1,\ldots,g_k\}$ from an $\eps$-sample $P$ of $M$ and a parameter
$r>0$ whose total length is within a factor 
of the total length of a shortest basis 
in $\homo(M)$. The factor depends on $\eps$, $r$, and $\rho(M)$.

\begin{theorem}
Let $M\subset \mathbb{R}^d$ be a smooth, closed  manifold 
with $\ell$ as the length
of a shortest basis of $\homo(M)$. 
Given a set $P\subset M$ of $n$ points
which is an $\eps$-sample of $M$ and
$4\eps\leq r \leq \min\{\frac{1}{2}\sqrt{\frac{3}{5}}\rho(M), 
\rho_c(M) \}$,
one can compute a set of loops $G$
in $O(n(n+n_e)^2(n_e+n_t))$ time where 
$$\frac{1}{1+\frac{4r^2}{3\rho^2(M)}}\ell\leq\length(G)\leq 
(1+\frac{4\eps}{r})\ell, \mbox{ and}$$ $n_e$ and $n_t$ are the
numbers of edges and triangles respectively 
in the Rips complex $\Rips^{2r}(P)$.
\label{thm:main}
\end{theorem}

The above result suggests that
$\lim_{\frac{\eps}{r},r\rightarrow 0}\length(G)\rightarrow \ell$.
To make $\frac{\eps}{r}$ and $r$ simultaneously 
approaching $0$, one may take $r=O(\sqrt \eps)$ and let $\eps\rightarrow 0$. 
We note that $n_e=O(n^2)$ and $n_t=O(n^3)$
giving an $O(n^{8})$ worst-case complexity
for the algorithm. However, if $r=\Theta(\eps)$
and points in $P$ have $\Omega(\eps)$ pairwise distance,
$n_e$ and $n_t$ reduce to $O(n)$ by a result of~\cite{CO08}.
In this case we get a time complexity of $O(n^4)$.
In arriving at Theorem~\ref{thm:main}, we also prove the following
result which is of independent interest.

\begin{theorem}
Let $\KK$ be a finite simplicial complex with non-negative 
weights on edges. A shortest basis
for $\homo(\KK)$ can be computed in $O(n^4)$ time
where $n$ is the size of $\KK$.
\label{thm:main2}
\end{theorem}

\section{Algorithm description} 
The algorithm that we propose proceeds as follows. 
We compute a Rips complex $\Rips^{2r}(P)$
out of the given point cloud $P\subset M$.
Next, we compute the rank $k$ of $\homo(M)$ by considering
the persistent homology group
\[
\homo^{r,2r}(\Rips(P)) = \mbox{ image } \iota_*
\]
where the inclusion $\iota: \Rips^{r}(P)\hookrightarrow
\Rips^{2r}(P)$ induces the homomorphism
$\iota_*: \homo(\Rips^{r}(P))
\rightarrow \homo(\Rips^{2r}(P))$. 
It is known that the rank of $\homo^{r,2r}(\Rips(P))$ 
coincides with that
of $\homo(M)$ for appropriate $r$.

We show that a shortest basis of $\homo^{r,2r}(\Rips(P))$
approximates a shortest basis of $\homo(M)$. Therefore, we aim 
to compute a shortest basis of $\homo^{r,2r}(\Rips(P))$
from $\Rips^{r}(P)$ and $\Rips^{2r}(P)$. To accomplish this,
the algorithm augments $\Rips^{2r}(P)$ by putting a weight $w(e)$
on each edge $e\in\Rips^{2r}(P)$. The weights are of two types:
either they are the lengths of the edges, or a very large value $W$
which is larger than $k$ times the total weight of $\Rips^{r}(P)$.
Precisely we set 
\[
w(e)= \left\{\begin{array}{ll}
 \mbox{length of $e$} & \mbox{if $e\in \Rips^{r}(P)$}\\
 W & \mbox{if $e\in\Rips^{2r}(P)\setminus \Rips^{r}(P)$}.
\end{array}
\right.
\]
Let the complex $\Rips^{2r}(P)$ augmented with weights be
denoted as $\Rips^{2r +}(P)$.
A shortest basis of $\homo(\Rips^{2r+}(P))$ does not
necessarily form a shortest basis of $\homo^{r,2r}(\Rips(P))$.
However, the first $k$ loops sorted according to lengths in 
a shortest basis of $\homo(\Rips^{2r+}(P))$ form a 
shortest basis of $\homo^{r,2r}(\Rips(P))$. We
give an algorithm to compute a shortest basis for any
simplicial complex which we apply to $\Rips^{2r+}(P)$.


Since we are interested in computing the generators of the first
homology group, it is sufficient to consider all simplices up to dimension
two, that is, only vertices, edges, and triangles in the simplicial
complexes that we deal with. Henceforth, we assume that all complexes
that we consider have simplices up to dimension two.

\subsection{Computing loops}
We will prove later that a shortest basis for
$\homo^{r,2r}(\Rips(P))$ indeed approximates a shortest basis for
$\homo(M)$. The algorithm \mainAlg computes them.
\begin{algorithm}[h!]
\floatname{algorithm}{Algorithm}
\caption{ {\mainAlg($P,r$)} }
\begin{algorithmic}[1]
   \STATE Compute the Rips complex $\Rips^{2r}(P)$
	and a weighted complex $\Rips^{2r+}(P)$ 
from it as described.
     \STATE Compute the rank $k$ of $\homo^{r,2r}(\Rips(P))$
	by the persistence algorithm.
	\STATE  Compute a shortest basis for $\homo(\Rips^{2r+}(P))$.	
       \STATE Return the first $k$ smallest loops from this shortest basis.
\end{algorithmic}
\end{algorithm}

\begin{theorem}
The algorithm {\sc ShortLoop}$(P,r)$ computes a shortest basis 
for the persistent homology group 
$\homo^{r,2r}(\Rips(P))$. 
\label{thm:SPgenerators}
\end{theorem}
\begin{proof}
Let $g_1,\ldots,g_{a}$ be the set of generators
sorted according to the non-decreasing lengths which
are computed in step 3. They generate $\homo(\Rips^{2r+}(P))$.
Out of these generators the algorithm outputs the
first $k$ generators $g_1,\ldots,g_k$. 
Since $k$ is the rank of $\homo^{r,2r}(P)$
there are $k$ independent generators in $\homo(\Rips^{r}(P))$
which remain independent in $\homo(\Rips^{2r+}(P))$.
We claim that the loops $g_1,\ldots,g_k$ reside in $\Rips^{r}(P)$.
For if they do not, the sum of their lengths would be more
than $W$ which is $k$ times larger than the
total weight of $\Rips^{r}(P)$.
Then, we can argue that any $k$ independent set of loops from
$\Rips^{r}(P)$ which remain independent
in $\homo(\Rips^{2r+}(P))$ can replace $g_1,\ldots,g_k$ to have
a smaller length so that $g_1,\ldots,g_{a}$ could not
be a shortest basis of $\homo(\Rips^{2r+}(P))$.

The above argument implies that $g_1,\ldots,g_k$ is a basis of
$\homo^{r,2r}(P)$. If it is not a shortest basis,
it can be replaced by a shorter one so that again we would
have a basis of $\homo(\Rips^{2r+}(P))$ which is shorter than the
one computed. This is a contradiction.
\end{proof}\\

It remains to show how to 
compute a shortest basis of $\homo(\Rips^{2r+}(P))$
in step 3 of \mainAlg.

\subsection{Shortest basis}
\label{sploop-sec}
Let $\KK$ be any finite simplicial complex embedded in $\mathbb{R}^d$
whose edges have non-negative weights.
To compute a shortest basis for $\homo(\KK)$ we make use of the fact that
$\homo(\KK)$ is a vector space as we restrict ourselves to $\mathbb{Z}_2$ 
coefficients. For such cases, Erickson and Whittlesey~\cite{EW05}
observed that if a set of loops $\mathcal{L}$ in $\KK$ 
contains a shortest basis, then
the greedy set $G$ chosen from $\mathcal{L}$ is a shortest basis. 
The greedy set 
$G$ of $\mathcal{L}$
is an {\it ordered} set of loops $\{g_1,\ldots \g_k\}$, 
$k={\mathrm rank}\,\homo(\KK)$, satisfying the following
condition. The first element $g_1$ is 
the shortest loop in $\mathcal{L}$ which is nontrivial in $\homo(\KK)$. Suppose $g_1, \ldots, \g_i$ 
have already been defined in the set $G$. 
The next chosen loop $\g_{i+1}$ is the shortest loop
in $\mathcal{L}$ which is independent of $g_1,\ldots, g_i$, that is, $[\g_{i+1}]$ cannot be written 
as a linear combination of $[\g_1],...,[\g_i]$.
The check for independence is a costly step in this greedy
algorithm which we aim to reduce.
We construct a set of {\em canonical} loops which contains
a basis of $\homo(\KK)$. This set is pruned by 
a persistence based algorithm before applying the
greedy algorithm. 

\subsubsection{Canonical loops}
We start with citing a result of Erickson and Whittlesey~\cite{EW05}.
A simple cycle $L$ is {\em tight} if it 
contains a shortest path between every pair 
of points in $L$.  
\begin{proposition}
With non-negative weights, every loop in a 
shortest basis of $\homo(\KK)$ is tight.
\label{lem:tight}
\end{proposition}
To collect all tight loops, we consider the canonical 
loops defined as follows. Let 
$T$ be a {\em shortest path tree} in  $\KK$ rooted at $p$. Notice that
we are not assuming $T$ to be unique but it is fixed once computed. For any two nodes $q_1,q_2\in P$, let $\SP_T(q_1,q_2)$ 
denote the unique path from $q_1$ to $q_2$ in $T$. Let $E_{T}$ be the set of edges in $T$. 
Given a non-tree edge $e = (q_1, q_2) \in E\setminus E_{T}$, 
define the \emph{canonical loop} of $e$ with respect to $p$, $c_p(e)$ in short, 
as the loop formed by 
concatenating $\SP_T(p,q_1)$, $e$, and $\SP_T(q_2,p)$, that is,
$$c_p(e)=\SP_T(p,q_1) \concatenate e \concatenate \SP_T(q_2,p).$$
Let $C_p$ be the set of all canonical loops with respect to $p$, 
i.e., $C_p = \{c_p(e): e\in E\setminus E_T\}$.
Then we have the following easy consequence.
\begin{proposition}
$\cup_{p\in P}C_p$ contains all tight loops.
\label{lem:tight1}
\end{proposition}
Therefore $\cup_{p\in P}C_p$ is a set 
of loops from which the greedy set can be selected. 
However, $\cup_{p\in P}C_p$ can be a very large set containing  
possibly many trivial loops which result into many unnecessary
independence checks. 
To remedy this, we identify the greedy set $G_p$ of $C_p$ and choose the
greedy set from the union $\cup_{p\in P}G_p$ instead of $\cup_{p\in P}C_p$. 
It turns out that $G_p$ can be computed by a persistence based
algorithm thereby avoiding explicit independence checks.

If the lengths of the loops in $C_p$ are distinct, the greedy set $G_p$
is unique. However, in presence of equal length loops
we need a mechanism to break ties. 
For this we introduce the notion of {\it canonical order}.  
We assign a unique number $\nu(e)$ between $1$ to $m$ to each 
non-tree edge $e$ if there are $m$ of them. 
For any two non-tree edges $e$ and $e'$, let 
$e < e'$ if and only if either $\length(c_p(e))< \length(c_p(e'))$,
or $\length(c_p(e))=\length(c_p(e'))$ and $\nu(e)< \nu(e')$.
The total order imposed by `$<$' provides the canonical order
\[
e_1 < e_2<\ldots < e_m.
\]
Based on this canonical order, we form the greedy set $G_p$ of $C_p$
as described in the beginning of Section~\ref{sploop-sec}. 

Below we argue that $\cup_{p\in P}G_p$ is good for our purpose 
and each set $G_p$ can be computed based on the persistence algorithm.  

\begin{proposition}
The greedy set chosen from $\cup_{p\in P}G_p$ is a shortest basis 
of $\homo(\KK)$.
\label{greedy}
\end{proposition}
\begin{proof}
We show that $\cup_{p\in P}G_p$ contains a shortest basis of $\homo(\KK)$.
Then, the proposition
follows by the argument as delineated at the beginning 
of section~\ref{sploop-sec}.

Consider all canonical loops $\cup_{p\in P}C_p$. Sort them in
non-decreasing order of their lengths. 
If two loops have equal lengths and if there are 
points $p_i\in P$ for which both of them are in $C_{p_i}$, 
break the tie using the canonical order applied to the
canonical loops for any such one point. 
Based on this order let $G$ be the greedy set from $\cup_{p\in P}C_p$. 
Proposition~\ref{lem:tight}
and Proposition~\ref{lem:tight1} imply that 
$\cup_{p\in P}C_p$ contains 
a shortest basis of $\homo(\KK)$ and thus $G$ is a shortest basis. 
Consider any loop $L$
in $G$. It is a canonical loop with respect to 
some $q\in P$ for which all loops appearing
before $L$ in the canonical order precede
it in the sorted sequence. The loop $L$ is independent of 
the loops in $\cup_{p\in P}C_p$ appearing before $L$, 
in particular independent of the loops in $C_q$ 
appearing before $L$ 
in the canonical order, which means $L \in G_q$. 
Therefore $\cup_{p\in P}G_p$ contains a shortest basis $G$ of $\homo(\KK)$. 
The proposition follows. 
\end{proof}

Motivated by the above observations, we formulate an algorithm
\CanonGen that computes the greedy set $G_p$ of $C_p$.
We note that, very recently, Chen and Freedman~\cite{CF09} proposed a similar
algorithm which computes
an {\em approximation} of a shortest basis of a simplicial
complex rather than an optimal one.

\begin{algorithm}[h!]
\floatname{algorithm}{Algorithm}
\caption{  \CanonGen($p$, $\KK$) } 
\begin{algorithmic}[1]
   \STATE  Construct a shortest path tree $T$ in $\KK$
with $p$ as the root. 
Let $E_T$ denote the set of tree edges.
	\STATE  For each non-tree edge $e = (q_1, q_2) \in E \setminus E_T$, 
let $c_p(e)$ be the canonical loop of $e$. 
	\STATE  Perform the persistence algorithm based on the following 
filtration of $\KK$: all the vertices in $P={\rm Vert}(\KK)$, followed by all tree edges in 
$T$, followed by non-tree edges in the 
{\em canonical order}, and followed by all the triangles in $\KK$. 
There are $k = \rank(H_1(\KK))$ number of edges 
unpaired after the algorithm, and each of them is necessarily a non-tree edge. 
Return the set of canonical loops associated with them. 
\end{algorithmic}
\end{algorithm}

\begin{proposition}
\CanonGen$(p,\KK)$ outputs the greedy set $G_p$ chosen from $C_p$.
\label{lem:SLfixedbasept}
\end{proposition}
\begin{proof}
Let $\{e_1, e_2 \cdots, e_m\}$ be the non-tree edges in the shortest path
tree $T$ listed in the canonical order. Let $G_p=\{c_p(e_1^*), c_p(e_2^*), 
\cdots, c_p(e_k^*)\}$. 
It suffices to show that $\{e_1^*, e_2^* \cdots, e_k^*\}$ is the set 
of unpaired 
edges. Observe that for any $e_i^*$, $c_p(e_i^*)$ 
is independent of any subset of $\{c_p(e_j): e_j < e_i^*\}$. 

We prove the proposition by contradiction. Assume some $e_i^*$ gets paired
by a triangle $t$ in the persistence algorithm. Let $\KK_t$ denote the
complex in the filtration right before $t$ is added. 
Let $f: \KK_t\hookrightarrow \KK$ 
be the inclusion map; it induces a homomorphism $f_* = \homo(\KK_t) \rightarrow \homo(\KK)$. 
Let $[L]_t$ denote the homology class in $\KK_t$ carried by the loop $L$. 
The boundary
$\partial t$ uniquely determines a subset of unpaired 
positive edges $e_1'< \cdots < e_n'$ 
in $\KK_t$ such that $[\partial t]_t = [c_p(e_1')]_t +\cdots+ [c_p(e_n')]_t$. 
The persistence
algorithm~\cite{ELZ02} picks the youngest one from this subset 
to pair with $t$, i.e., $e_i^* = e_n'$. 
On the other hand, we have 
$$[c_p(e_1')] + \cdots + [c_p(e_{n-1}')] + [c_p(e_i^*)] = f_*( [c_p(e_1')]_t + \cdots + [c_p(e_{n-1}')]_t + [c_p(e_i^*)]_t) = f_*([\partial t]_t)=0 $$
which means that $c_p(e_i^*)$ is dependent on a 
subset of $\{c_p(e_j): e_j < e_i^*\}$. We reach
a contradiction. 
\end{proof}

All previous results put together provide a 
greedy algorithm for computing a shortest basis 
of $\homo(\KK)$.

\begin{algorithm}[h!]
\floatname{algorithm}{Algorithm}
\caption{  \SPGen($\KK$) } 
\begin{algorithmic}[1]
   \STATE For each $p\in P={\rm Vert}(\KK)$ compute $G_p:=$\CanonGen($p$,$\KK$).  Let $k=|G_p|$. 
	\STATE Sort all loops in $\cup_p G_p$ by their lengths in the increasing order.  Let $\g_1,\ldots,\g_{k|P|}$ be this sorted list. 
	\STATE  Initialize $G:=\{g_1\}$.
	\FOR{$i:=2$ to $k|P|$,}
		\IF{$|G| = k$,}
			\STATE Exit the for loop.
		\ELSIF{$\g_i$ is independent of all loops in $G$,}
		 	\STATE Add $g_i$ to $G$.
		\ENDIF
	\ENDFOR	
	\STATE Return $G$.
\end{algorithmic}
\end{algorithm}

\subsubsection{Checking independence}
In step 7 of {\sc SPGen} we need to determine if a generator $g$ is
independent of all generators $g_1',\ldots,g_s'$ so far selected in $G$. 
We obtain $g$ from running persistence
algorithm on a shortest path tree based filtration for a point $p$ in step 3
of {\sc CanonGen}. At the end of this persistence algorithm
we must have gotten an unpaired edge, say $e$, where $c_p(e)=g$. 
To determine if $g$ is independent of all generators 
selected so far we adopt a sealing technique proposed in~\cite{CF09}.
We fill $g_1'\ldots g_s'$ with triangles.
The filling is done only 
combinatorially by choosing
a dummy vertex, say $v$, and
adding triangles $vv_iv_{i+1}$ for each edge $v_iv_{i+1}$ of the
loops to be filled. 
Let $\KK'$ be the new complex after adding these triangles and their edges 
to $\KK$. In effect, these triangles and edges destroy the
generators $g_1',\ldots,g_s'$ from $\KK$. They destroy the
generator $g$ as well if and only if $g$ is dependent on
$g_1',\ldots, g_s'$. Since we are sealing according to
the greedy order, the proof of Lemma
4.4 in~\cite{CF09} applies to establish this fact.
Whether $g$ is rendered trivial or not
can be determined as follows. 
We continue the persistence algorithm
corresponding to the vertex $p$ with the addition of
the simplices in $\KK'\setminus \KK$ and check if $e$ is now paired or not. 

Let $n_v$, $n_e$, and $n_t$ denote the number of
vertices, edges, and triangles respectively in $\KK$.
Notice that we add at most $n_e$ edges and triangles for sealing
since the dummy vertex is added to at most $n_e$ edges
to create new triangles in $\KK'$. 

\subsection{Time complexity}
\label{time}

First, we analyze the time complexity of {\sc CanonGen}.
Shortest path tree computation
in step 1 of {\sc CanonGen} takes $O(n_v\log n_v+n_e)$ time.
The persistence algorithm for {\sc CanonGen} can be
implemented using matrix reductions~\cite{CEM06} in 
time $O((n_v+n_e)^2(n_e+n_t))$. This is because
there are $n_v+n_e$ rows in this matrix and each
insertion of $n_e+n_t$ simplices can be implemented
in $O(n_v+n_e)$ column operations each taking 
$O(n_v+n_e)$ time. Therefore, {\sc CanonGen}
takes $O(n_v\log n_v + (n_v+n_e)^2(n_e+n_t))$ time.

Step 1 of {\sc SPGen} calls
{\sc CanonGen} $n_v$ times. 
Therefore, step 1 of {\sc SPGen} takes $O(n_v^2\log n_v + n_v(n_v+n_e)^2
(n_e+n_t))$ time.
Step 2 of {\sc SPGen} can be performed in $O(n_vk\log n_vk)$ time
where $k=O(n_e)$ is the rank of $\homo(\KK)$.
The time complexity for independence check in step 7
is dominated by the persistence algorithm which is continued
on $\KK$ to accommodate simplices in $\KK'$.
Since we add $O(n_e)$ new simplices in $\KK'$, it has the 
same asymptotic complexity as for running the
persistence algorithm on $\KK$. We conclude that 
{\sc SPGen} spends 
$O(n_v(n_v+n_e)^2(n_e+n_t))$ time in total.
If we take $n=|\KK|$, this gives an $O(n^4)$ time complexity.

Now, we analyze the time complexity of {\sc ShortLoop} which is
the main algorithm. Let $n_e$ and $n_t$ be the
number of edges and triangles in $\Rips^{2r}(P)$
created out of $n$ points. Step 1 takes at most $O(n+n_e+n_t)$ time since
we only compute edges and triangles of
$\Rips^{2r}(P)$ out of $n$ points.
Accounting for the persistence algorithm in step 2
and the time complexity of step 3 
we get that {\sc ShortLoop}
takes 
\[
O(n(n+n_e)^2(n_e+n_t))
\mbox{ time}. 
\]

The procedure {\sc SPGen}($\KK$) 
computes canonical sets $G_p$ which
is ensured by Proposition~\ref{lem:SLfixedbasept}. Then, it forms a greedy set
from these canonical sets which is a shortest basis for
$\homo(\KK)$ by Proposition~\ref{greedy}. This and the 
time analysis for {\sc SPGen} establish Theorem~\ref{thm:main2}.

\section{Approximation for $M$}
The algorithm {\sc SpGen} is used in {\sc ShortLoop} to produce
a shortest basis for the persistent homology group $\homo^{r,2r}(\Rips(P))$.
Proposition~\ref{cech-rips-gen} in this section
shows that a shortest basis
of $\homo^{r,2r}(\Rips(P))$ coincides with a shortest basis
in $\homo(\Cech^r(P))$. Therefore, if we show that a shortest basis
in $\homo(\Cech^r(P))$ approximates  
a shortest basis in $\homo(M)$, we have the approximation result 
of Theoerm~\ref{thm:main}.
 
\subsection{Connecting $M$, \v{C}ech complex, and Rips complex}
First, we note the following result established in~\cite{NSW06}
which connects $M$ with the union of the balls $P^r = \cup_{p\in P}B(p, r)$.
\begin{proposition}
Let $P\subset M$ be an $\eps$-sample. 
If $2\eps\leq r \leq \sqrt{\frac{3}{5}}\rho(M)$, 
there is a deformation retraction from $P^r$ to $M$ 
so that the corresponding retraction
$t: P^r\rightarrow M$ has $t(B)\subset B$ 
for any ball $B \in \{B(p,r)\}_{p\in P}$. 
\label{cover-mani}
\end{proposition}

Recall that $\Cech^{2r}(P)$ is the nerve of the cover
$\{B(p,r)\}_{p\in P}$ of the space $P^{r}$.  By a result of
Leray~\cite{Leray}, it is known that $P^{r}$ and $\Cech^{2r}(P)$ are
homotopy equivalent. The next proposition follows from examining the specific
equivalence maps used to prove the Nerve Lemma in Hatcher~\cite{Hatcher}.
In particular, the simplices of the \v{C}ech complex are mapped to a subset of 
the union of the balls centered at their vertices, see Appendix for its proof. 
\begin{proposition}
There exists a homotopy equivalence $f\colon \Cech^{2r}(P) 
\rightarrow P^{r}$ such that
for each simplex $\sigma\in \Cech^{2r}(P)$, 
one has $f(\sigma)\subset \cup_{p\in \Vert(\sigma)}B(p,r)$ 
and $f(p) = p$ for any $p\in P$.
\label{cech-cover} 
\end{proposition}

The two propositions above together provide the connection between $M$ and the \v{C}ech complex:

\begin{proposition}
Let $P\subset M$ be an $\eps$-sample. If $2\eps\leq r \leq \sqrt{\frac{3}{5}}\rho(M)$, there is a homotopy equivalence 
map  $h=t\circ f : \Cech^{2r}(P) \rightarrow M$ 
such that $h(\sigma)\subset M\cap (\cup_{p\in \Vert(\sigma)} B(p,r))$ 
and $h(p) = p$ for any $p\in P$.
\label{cech-mani}
\end{proposition}

Now we establish a  connection between \v{C}ech complex and Rips complexes
which helps proving Proposition~\ref{cech-rips-gen}. 
\begin{proposition} 
Let $P\subset M$ be an $\eps$-sample. Then, for
$4\eps\leq r \leq \frac{1}{2}\sqrt{\frac{3}{5}}\rho(M)$,
\begin{equation*}
\homo^{r,2r}(\Rips(P)) \hsim \homo(\Cech^{r}(P)) \stackrel{j_{1*}}{\hsim} \homo(\Cech^{2r}(P)) \stackrel{j_{2*}}{\hsim}  \homo(\Cech^{4r}(P)), 
\end{equation*}
where $j_{1*}$ and $j_{2*}$ are induced by the inclusion maps $j_1$ and $j_2$ respectively. 
Moreover, if 
\begin{equation*}
\Cech^{r}(P) \stackrel{i_{1}}{\hookrightarrow} \Rips^{r}(P)) \stackrel{i_{2}}{\hookrightarrow} \Cech^{2r}(P)) \stackrel{i_{3}}{\hookrightarrow} \Rips^{2r}(P)) \stackrel{i_{4}}{\hookrightarrow} \Cech^{4r}(P), 
\end{equation*}
then $j_1 = i_2\circ i_1$, and $j_2 =i_4 \circ i_3$ and $\homo^{r,2r}(\Rips(P))=image~(\iota_*)$ where
$\iota_*:\homo(\Rips^{r}(P))\rightarrow \homo(\Rips^{2r}(P))$ is induced by the inclusion $\iota=i_3\circ i_2$. 
\label{cech-rips}
\end{proposition}
\begin{proof}
Based on Proposition~\ref{cech-mani}, it can be proved by following the 
idea in~\cite{CO08} of intertwined \v{C}ech and Rips complexes. 
\end{proof}

By definition the set of edges in
$\Cech^{r}(P)$ is same as the set of edges in $\Rips^{r}(P)$.  This
means a set of loops in $\Rips^{r}(P)$ also forms a set of loops in
$\Cech^{r}(P)$. In light of Proposition~\ref{cech-rips}, this implies:

\begin{proposition}
Let $P\subset M$ be an $\eps$-sample and $4\eps\leq r \leq \frac{1}{2}\sqrt{\frac{3}{5}}\rho(M)$. 
Then $\homo^{r,2r}(\Rips(P)) \hsim \homo(M)$ and a basis 
for $\homo^{r,2r}(\Rips(P))$ is 
shortest if and only if it is shortest for $\homo(\Cech^{r}(P))$. 
\label{cech-rips-gen}
\end{proposition}
\begin{proof}
From Proposition~\ref{cech-mani}
and Proposition~\ref{cech-rips}, we have 
$$\homo^{r,2r}(\Rips(P)) \hsim \homo(\Cech^{r}(P)) \hsim  \homo(M).$$
Let $A=\{a_1, \cdots, a_k\}$ be a shortest basis for 
$\homo^{r,2r}(\Rips(P))$. 
Each $a_i$ is a loop in $\Rips^{r}(P)$ and hence in $\Cech^{r}(P)$. 
Obviously $A$ is a basis of $\homo(\Cech^{r}(P))$ as the 
inclusion map from $\Cech^{r}(P)$
to $\Rips^{r}(P)$ induces a homomorphism. 
Thus, a shortest basis for $\homo(\Cech^{r}(P))$ 
must be no longer than that of $\homo^{r,2r}(\Rips(P))$. 
Similarly if  
$A=\{a_1, \cdots, a_k\}$ is a shortest basis of $\homo(\Cech^{r}(P))$, 
then each $a_i$ must be in $\Rips^{r}(P)$ and survive 
in $\Rips^{2r}(P)$ as it must survive in $\Cech^{4r}(P)$.
Thus $A$ is a basis for $\homo^{r,2r}(\Rips(P))$ 
and hence a shortest basis of $\homo^{r,2r}(\Rips(P))$
is no longer than that of $\homo(\Cech^{r}(P))$. 
This proves the proposition. 
\end{proof}


\subsection{Bounding the lengths}
Our idea is to argue that a shortest basis of 
$\homo(\Cech^{r}(P))$ can be pulled back to a basis of $\homo(M)$ 
by the map $h$
of Proposition~\ref{cech-mani}. We argue that the lengths of the
generators cannot change too much in the process.

Let $g$ be any closed curve in $M$. Following~\cite{BSLT00}, we define a
procedure to approximate $g$ by a loop $\hat{g}$ in the $1$-skeleton of
$\Cech^{r}(P)$. This procedure called {\em Decomposition method} is not
part of our algorithm, but is used in our argument about length
approximations of loops in $M$.

\noindent
\paragraph{Decomposition method:}
If $\ell=\length(g) > r-2\eps > 0$, we can write
$\ell=\ell_0+(\ell_1+\ell_1+\ldots +\ell_1) + \ell_0$ where
$\ell_1=r-2\eps$ and $r-2\eps> \ell_0 \geq (r-2\eps)/2$.
Starting from an arbitrary point, say $x$,
split $\g$ into pieces whose lengths coincide with the decomposition
of $\ell$. This produces a sequence of points $x=x_0,x_1,\ldots, x_m=x$
along $\g$ which divide it according to the lengths constraints.
Because of our sampling condition, each point $x_i$ has a point
$p_i\in P$ within $\eps$ distance. We define a loop
$\hat{g}=\{p_0p_1\ldots p_m\}$
with consecutive points joined by line segments.
Proposition~\ref{prop:lbound} shows that $\hat{g}$ resides
in the $1$-skeleton of $\Cech^{r}(P)$ (proof in the Appendix).

\begin{proposition}
Given a closed curve $\g$ on $M$ with $\length(g)> r-2\eps > 0$, 
{\em Decomposition method} finds a loop $\hat{g}$ from the $1$-skeleton 
of $\Cech^{r}(P)$ 
such that: 
$\length(\hat{g}) \le \frac{r}{r -2\eps} \length(\g)$. 
\label{prop:lbound}
\end{proposition}

Consider a basis of $\homo(M)$ where each generator is a closed
geodesic on $M$. For a smooth, compact manifold such a basis always
exists by a well known result in differential geometry~\cite{docarmo}.
Let $G=\{\g_1,\ldots,\g_k\}$ be this set of geodesic 
loops. By Proposition~\ref{prop:lbound}, we claim that there is a set of
loops $\hat{G}=\{\hat\g_1,\ldots,\hat\g_k\}$ in $\Cech^{r}(P)$
whose length is within a small factor of the length of $G$.
However, we need to show that $\hat{G}$ indeed generates 
$\homo(\Cech^{r}(P))$.
We show this by mapping each $\hat g_j\in \hat{ G}$
to $M$ by the homotopy equivalence $h$ and arguing that
$[h(\hat\g_j)]=[\g_j]$ in $\homo(M)$.
Since $h$ is a homotopy equivalence map, it follows that 
the isomorphism $h^*:\homo(\Cech^{r}(P))\rightarrow \homo(M)$
maps the class $[\hat g_j]$ to $[\g_j]$. This implies that
$\hat {G}$ generates $\homo(\Cech^{r}(P))$.

To prove that $h(\hat g_j)$ is a representative of the class
$[g_j]$, we consider a tubular neighborhood of $g_j$
of radius $r$ which is smaller than the convexity
radius $\rho_c(M)$. Then, we show that each segment $p_ip_{i+1}$
of $\hat g_j$ is mapped to a curve $h(p_ip_{i+1})$ which
lies within this tubular neighborhood. Because of this
containment, $h(p_ip_{i+1})$ must be
homotopic to the geodesic segment $\gamma(x_i,x_{i+1})$ of $g_j$. 
All these homotopies together provide
a homotopy between $h(g_j)$ and $g_j$.
First we show that the tubular neighborhood of a segment of $g_j$
that we consider is indeed simply connected (see the Appendix for proof).

\begin{proposition}
Let $\gamma=\gamma(p,q)$ be a minimizing geodesic between two points 
$p, q \in M$. Consider its
tubular neighborhood $\Tub_s(\gamma)$ on $M$ that
consists of the points on $M$ within a
geodesic distance $s$ from $\gamma$, i.e.,
$\Tub_s(\gamma)=\{x\in M:  \min_{y\in \gamma} d_M(x, y)< s \}$. Then
if $s<\rho_c(M)$, $\Tub_s(\gamma)$ is contractible, 
in particular, $\Tub_s(\gamma)$ is simply connected.
\label{tub-prop}
\end{proposition}

\begin{proposition}
Let $P\subset M$ be an $\eps$-sample and 
$4\eps\leq r \leq \min\{\frac{1}{2}\rho(M), \rho_c(M)\}$. 
If $\hat g$ is the loop on $\Cech^{r}(P)$ constructed from a geodesic loop $g$ in $M$ by {\em Decomposition method}, 
then $[h(\hat g)] = [g]$ where $h$ is the homotopy equivalence defined in Proposition~\ref{cech-mani}.
\label{piece-homotopy}
\end{proposition}
\begin{proof}
Since $g$ is a geodesic loop,
it follows from standard results in differential geometry~\cite{docarmo}
that $\length(g) > 2\rho_c(M)$. 
Thus $\hat g$ can be constructed from a geodesic loop $g$ using {\em Decomposition method}.
Each vertex $p_i$ of $\hat g$ is within an $\eps$ Euclidean distance 
from the point $x_i$ in $g$. Next, notice that, since
$\Cech^r(P)$ uses balls of radius $r/2$, the stated range
of $r$ satisfies the condition of Proposition~\ref{cech-mani}.
By Proposition~\ref{cech-mani},
for any point $y$ on the segment $p_ip_{i+1}$, $h(y)$ is within
$r/2$ Euclidean distance to either $p_i$ or $p_{i+1}$. 
This implies that $h(y)$ is within $r/2+\eps$ Euclidean distance, 
and hence, by Proposition~\ref{euclid-geod}, within
$r$ geodesic distance 
to either $x_i$ or $x_{i+1}$.  In addition, since the sub-curve of the 
geodesic loop $g$ between $x_i$ and $x_{i+1}$, denoted $g(x_i, x_{i+1})$, is of length 
$\ell_1=r - 2\eps < \rho_c(M) $,  $g(x_i, x_{i+1})$ is the a minimizing geodesic between
$x_i$ and $x_{i+1}$. Therefore $h(p_ip_{i+1})\in  \Tub_{r}(\gamma(x_i,x_{i+1}))$ 
In particular, the geodesics $\gamma(x_i,h(p_i))$ and 
$\gamma(x_{i+1},h(p_{i+1}))$ reside in $\Tub_{r}(\gamma(x_i,x_{i+1}))$.

Consider the loop formed by the three geodesic segments $\gamma(x_i,x_{i+1})$, 
$\gamma(x_i,h(p_i))$, 
$\gamma(x_{i+1},h(p_{i+1}))$, and the curve $h(p_ip_{i+1})$. From 
Proposition~\ref{tub-prop}, this cycle is contractible in $M$ 
as it resides in $\Tub_{r}(\gamma(x_i,x_{i+1}))$. In fact,
there is a homotopy $H_i$ that takes $h(p_ip_{i+1})$ 
to $\gamma(x_i,x_{i+1})$ while
$H_i$ keeps $h(p_i)$ and $h(p_{i+1})$ on the geodesics
$\gamma(x_i,p_i)$ and $\gamma(x_{i+1},p_{i+1})$ respectively. We can combine all
homotopies $H_i$ for $0\leq i \leq m$ to define a homotopy
between $h(\hat g)$ and $g$. It follows that $[h(\hat g)]=[g]$. 
\end{proof}

\begin{proposition}
Let $P\subset M$ be an $\eps$-sample and 
$4\eps\leq r \leq \min\{\frac{1}{2}\rho(M), \rho_c(M)\}$. 
If $G=\{g_1,\ldots,g_k\}$ and $G'=\{g_1',\ldots,g_k'\}$
are the generators of a shortest basis of 
$\homo(M)$ and $\homo(\Cech^{r}(P))$ respectively, 
then we have $\length(G') \leq (1+\frac{4\eps}{r})\length(G)$.
\label{MRips}
\end{proposition}
\begin{proof}
It is obvious that any $g_i$ must be a geodesic loop.  
Let $\hat g_i$ be the loop constructed 
by {\em Decomposition method} in the $1$-skeleton of $\Cech^{r}(P)$. 
Thus, we have
a set $\hat G=\{\hat g_1,\cdots, \hat g_k\}$. 
By Proposition~\ref{piece-homotopy}, there 
is a homotopy equivalence 
$h: \Cech^{r}(P)\rightarrow M$ so that $[h(\hat g_j)] = [g_i]$, 
which means that $\hat G$ is also a basis of $\homo(\Cech^{r}(P))$.  
By Proposition~\ref{prop:lbound}, 
$$\length(G') \leq \length(\hat G) \leq \frac{r}{r -2\eps} \length(G)  \leq  (1+\frac{4\eps}{r})\length(G).$$ 
\end{proof}

We now consider the opposite direction, and provide a lower bound 
for the total length of a shortest basis of $\homo(\Cech^{r}(P))$ 
in terms of the length of a shortest basis of $\homo(M)$. 

\newcommand{\geodesic}	{\gamma}

\begin{proposition}
Let $P\subset M$ be an $\eps$-sample and 
$4\eps\leq r \leq \min\{\frac{1}{2}\rho(M), \rho_c(M)\}$. 
Let $G$ and $G'$ be defined as in Proposition~\ref{MRips}.
We have $\length{G}\leq (1+\frac{4r^2}{3\rho^2(M)})\length(G')$.
\label{prop:fromRipstoM}
\end{proposition}
\begin{proof}
We construct a set of loops in $M$ from $G'$. First, we show that
the length of these loops is at most $(1+\frac{4r^2}{3\rho^2(M)})$ 
times the length
of $G'$. Next, we show that the constructed loops generate $\homo(M)$.

For each loop $g'\in G'$, we construct $\bar{g}$ as follows. The vertices 
and edges of $g'$ are the vertices and edges of $\Cech^r(P)$. 
For an edge $e = pq \in g'$, $p, q \in P$ thus $p, q \in M$. 
We connect $p$ and $q$ by the geodesic $\gamma(p,q)$ on $M$, 
and map $e$ to this geodesic. Mapping each edge in $g'$ on $M$, 
we obtain $\bar{g}$. Thus we obtain a set $\bar G = \{\bar g_1,\cdots, \bar g_k\}$. 
By Proposition~\ref{euclid-geod}, $d_M(p,q)\leq (1+\frac{4d^2(p, q)}{3\rho^2(M)})d(p, q)\leq
(1+\frac{4r^2}{3\rho^2(M)})d(p,q)$. Hence the length bound follows. 

We now show that the set $\bar G$ is a basis for $M$. 
Consider mapping $g_j'\in G'$ to $M$ by the equivalence map
$h$. Each edge $e = pq \in g_j'$ is mapped to a curve
$h(pq)$. From Proposition~\ref{cech-mani}, 
we have that $h(p) = p$ and $h(q) = q$
and each point of $h(pq)$ is within  $r/2$ Euclidean distance and hence
$r$ geodesic distance to either $p$ or $q$. This implies that
$h(pq) \subset \Tub_{r}(\gamma(p,q))$.  Then, by using similar argument
as in Proposition~\ref{piece-homotopy}, we claim that 
$\gamma(p,q)$ and $h(pq)$ are
homotopic. Combining all homotopies for each edge of $g_j'$,
we get that $h(g_j')$ is homotopic to $\bar{g_j}$.
Since $h$ is a homotopy equivalence, $h(G')$ 
and hence $\bar G=\{\bar{g_1},\ldots,\bar{g_k}\}$ are a basis of $\homo(M)$.
Therefore,
$$\length(G)\leq \length(\bar G) \leq (1+\frac{4r^2}{3\rho^2(M)}) \length(G').$$
\end{proof}

For an appropriate range of $r$, shortest bases in $\Cech^{r}(P)$ and 
$\homo^{r,2r}(\Rips(P))$ are same by Proposition~\ref{cech-rips-gen}.
\begin{theorem}
Let $P\subset M$ be an $\eps$-sample and $4\eps\leq r \leq \min\{\frac{1}{2}\sqrt{\frac{3}{5}}\rho(M), \rho_c(M) \}$. 
Let $G$ and $G'$ be a shortest basis of 
$\homo(M)$ and $\homo^{r,2r}(\Rips(P))$ respectively.
We have $\frac{1}{1+\frac{4r^2}{3\rho^2(M)}}\length(G)\leq \length(G')
\leq (1+\frac{4\eps}{r})\length(G)$.
\label{thm:lengthbound}
\end{theorem}

Theorem \ref{thm:main} follows from Theorem~\ref{thm:lengthbound}, 
Theorem \ref{thm:SPgenerators}, and the time complexity
analysis in section~\ref{time}. 

\section{Conclusions}
We have given a polynomial time algorithm for approximating a shortest
basis of the first homology group of a smooth manifold from a point data.
We have also presented an algorithm to compute a shortest 
basis for the first homology of any finite simplicial complex.
The question of computing a shortest basis for other homology
groups under $\mathbb{Z}_2$
has been recently settled by Chen and Freedman~\cite{CF10}
who show it to be NP-hard.

We use Rips complexes for computations 
and use \v{C}ech complexes for analysis.
One may observe that \v{C}ech complexes can be
used directly in the algorithm. Since we know that
$\Cech^{r}(P)$ is homotopy equivalent to $M$
for an appropriate range of $r$, we can
compute a shortest basis for $\homo(\Cech^{r}(P))$
which can be shown to approximate a shortest
basis for $\homo(M)$ using our analysis. In technical
terms, this will get rid of the weighting in step 1 and also step 4 of 
{\sc ShortLoop} algorithm, and make Theorem~\ref{thm:SPgenerators}
and Proposition~\ref{cech-rips-gen} redundant. Although the time
complexity does not get affected in the worst-case
sense, computing the triangles for \v{C}ech 
complexes becomes harder numerically in high dimensions
than those for the Rips complexes. This is
why we chose to describe an algorithm using 
the Rips complexes.\\


\section*{Appendix}

\myparagraph{Proof of Proposition~\ref{euclid-geod}.}
\begin{proof}
Let $\gamma(t)$ be the minimizing geodesic between $p$ and $q$
parameterized by arclength and set $l = d_M(p, q)$. By Proposition $6.3$
in~\cite{NSW06} we have that $l \le 2d(p,q)$. Let $u_t = \dot{\gamma}(t)$ be the
\emph{unit} tangent vector of $\gamma$ at $t$. We have $t = d_M(p, \gamma(t))$.

Let $B: T_{\gamma(t)} \times T_{\gamma(t)} \rightarrow T_{\gamma(t)}^\perp$ 
be the second
fundamental form associated with the manifold $M$. 
Since $\gamma$ is a geodesic,
$du_t / dt = B(u_t, u_t) = \ddot{\gamma}(t)$. 
Write $\rho=\rho(M)$ and $d=d(p,q)$ for convenience. 
From  Proposition $6.1$ in~\cite{NSW06}, 
we have the norm $\|\ddot{\gamma}(t)\| \leq 1/\rho$ as the norm of the second
fundamental form is bounded by $1/\rho$ in all directions, 
and thus $\|du_t / dt\| \leq 1/\rho$.
Hence we have that $$ \|u_t - u_p\| ~=~ \|\int_{[0,t]} du_y\| ~\leq~ \int_{[0,t]} \frac{1}{\rho}dy = \frac{t}{\rho}
~~~~\Rightarrow~~~~ \sin \frac{\angle (u_p, u_t)}{2} ~\le~ \frac{t}{2\rho}. $$

Furthermore, let $u \cdot v$ denote the dot-product between vectors
$u$ and $v$. Then we have that

\begin{equation*}
\begin{split}
\int_{[0, l]} u_t\cdot u_p~dt &= \int_{[0, l]} \cos \angle (u_t, u_p)~dt 
= \int_{[0, l]} ( 1 - 2\sin^2 \frac{\angle (u_t, u_p)}{2} ) dt \\
&\geq \int_{[0, l]} \left( 1-\frac{t^2}{2\rho^2}\right) dt 
= l - \frac{l^3}{6\rho^2} \\
\end{split}
\end{equation*}

On the other hand, observe that $\int_{[0, l]}
u_t \cdot u_p~dt$ measures the length of the (signed) projection of
$\gamma$ along the direction $u_p$. That is,
$$\int_{[0, l]} u_t \cdot u_p~dl_t = (q-p) \cdot u_p. $$
Hence we have that
\begin{eqnarray*}
d = \| p - q \| \geq (q-p)\cdot u_p \geq l - \frac{l^3}{6\rho^2} \Rightarrow l ~\le~ d + \frac{l^3}{6\rho^2} ~\le~ d + \frac{4d^3}{3\rho^2}~.  
\end{eqnarray*}

The last inequality follows from the fact that 
$l \le 2d$. This proves the lemma.
\end{proof}

\myparagraph{Proof of Proposition~\ref{prop:lbound}.}
\begin{proof}
From the construction and sampling condition,
it follows that, for $1\leq i \leq m-2$,
\begin{eqnarray*}
d(p_i,p_{i+1})&\leq& d(x_i,p_i)+d(x_i,x_{i+1})+d(x_{i+1},p_{i+1})\\
&<& 2\eps +\ell_1=r = \frac{r}{(r-2\eps)}\ell_1
\end{eqnarray*}
Similarly, 
\[
d(p_0,p_1)\leq \frac{r}{r-2\eps}\ell_0
\mbox{ and } 
d(p_{m-1},p_0) \leq \frac{r}{r-2\eps}\ell_0.
\]
Since $\frac{r}{r-2\eps}\ell_0 < r$, each edge
$p_ip_{i+1}$ belongs to $\Cech^{r}(P)$.
Therefore, we obtain a loop $\hat g=p_0p_1\ldots p_m$ in the
$1$-skeleton of $\Cech^{r}(P)$ whose length satisfies:
\begin{eqnarray*}
\length(\hat g)=\Sigma_{i=0}^{m-1}d(p_i,p_{i+1}) \leq 
\frac{r}{r-2\eps}\length(\g).
\end{eqnarray*}
\end{proof}\\

\myparagraph{Proof of Proposition~\ref{tub-prop}.}
\begin{proof}
We show that $\Tub_s(\gamma)$ deformation retracts to $\gamma$.
For any point $x\in \Tub_s(\gamma)$, consider a geodesic ball $B$
of radius $s$. Since $s$ is less than the convexity radius,
$\gamma\cap B$ has a unique point $x_m$ which is at a minimum geodesic
distance from $x$. 
Consider the retraction map $t:\Tub_s(\gamma) \rightarrow \gamma$
where $t(x)=x_m$. One can construct a deformation retraction
that deforms the identity on $\Tub_s(\gamma)$ to $t$ by moving
each point $x$ along the minimizing geodesic path that connect $x$ to 
$x_m$ in $\gamma$. 
\end{proof}

\myparagraph{Proof of Proposition~\ref{cech-cover}.}
\begin{proof}
The proof is based on that of Nerve Lemma in~\cite{Hatcher} (Chapter 4.G). 
Let $\Gamma$ be the barycentric subdivision of $\Cech^{2r}(P)$.
Taking the definitions of the maps $\Delta p$, $\Delta q$, and the
space $\Delta P^r$ from Hatcher~\cite{Hatcher}, we
consider the following sequence
\begin{equation}
\Cech^{2r}(P) \stackrel{h}{\leftrightarrow} \Gamma \stackrel[\Delta p]{\Delta q}{\rightleftarrows} \Delta P^r 
\stackrel{\pi}\rightarrow P^r.
\label{eq:f-ext}
\end{equation}
We prove the proposition by showing $f= \pi \circ \Delta q \circ h$ which
is a homotopy equivalence.
We first introduce the concept of mapping cylinder. For a map $f: X\rightarrow Y$, 
the mapping cylinder $M_f$ is the quotient space of the disjoint union $(X\times I)\bigsqcup Y$ with
$(x, 1)$ identified with $f(x)\in Y$, denoted $M_f = X\bigsqcup_f Y$, see Figure~\ref{mapping_cylinder}(a).
\begin{figure}[h!]
\begin{center}
\begin{tabular}{cc}
\includegraphics[width=0.5\textwidth]{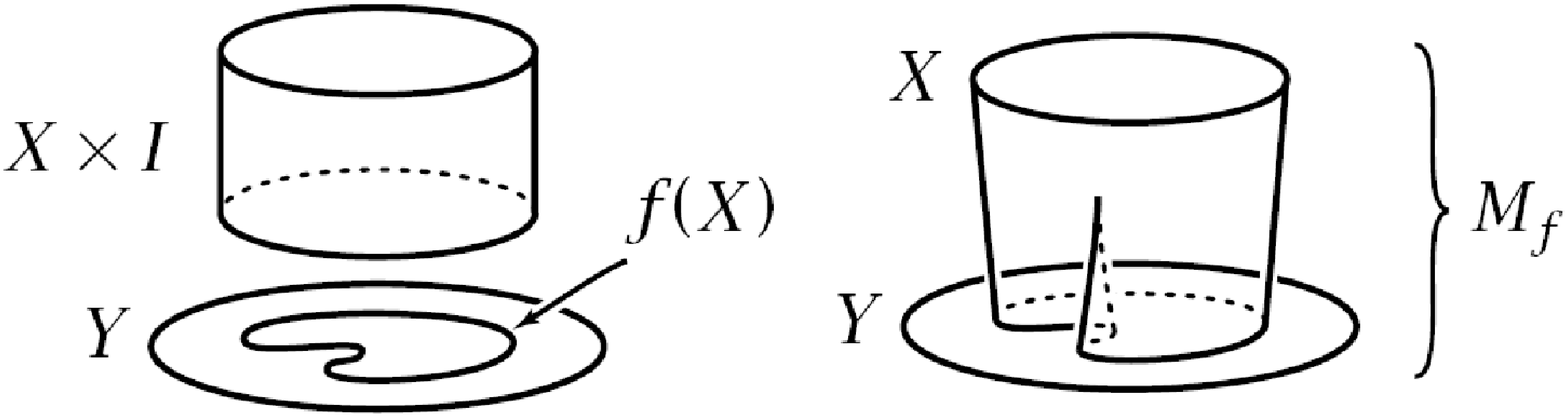} & \includegraphics[width=0.28\textwidth]{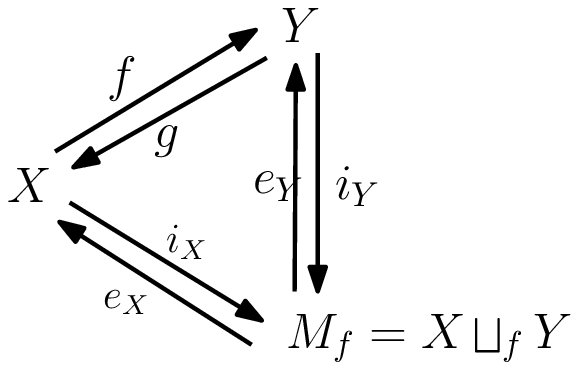} \\
(a) & (b) 
\end{tabular}
\end{center}
\caption{ (a) the mapping cylinder $M_f = X\bigsqcup_f Y$ (courtesy of Hatcher~\cite{Hatcher}); (b) the maps among $X$, $Y$ and $M_f$ 
\label{mapping_cylinder} }
\end{figure}
It is obvious that $M_f$ deformation retracts to $Y$. It is also well-known that $f$ is a homotopy equivalence map
if and only if $M_f$ deformation retracts to $X$, see Figure~\ref{mapping_cylinder}(b), where the map 
$g=e_X \circ i_Y$ is a homotopy equivalence map from $Y$ to $X$.

We are now ready to explain each map in the composition of the map $f$. 
$\Gamma$ is the barycentric subdivision of $\Cech^{2r}(P)$. 
Thus $h$ is an identity map between the underlying spaces 
of $\Cech^{2r}(P)$ and $\Gamma$.
Index the points in $P=\{p_i\}_{i=1}^m$ arbitrarily. Let $B_i = B(p_i, r)$.
To facilitate the argument, label the vertices in $\Gamma$ using $B_i$'s and their finite intersections, 
see Figure~\ref{gamma_delta_u}. Each edge (one simplex) in $\Gamma$ is associated with an inclusion map, 
which induces a sequence of inclusion maps over a simplex of any dimension in $\Gamma$.

\begin{figure}
\begin{center}
\includegraphics[width=0.7\textwidth]{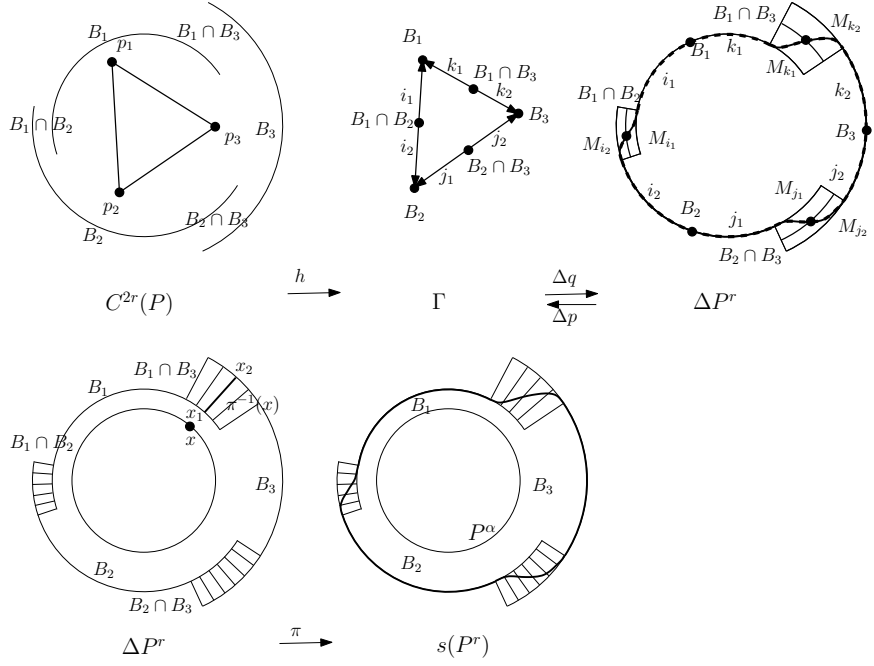}
\end{center}
\caption{Illustration of the maps and the spaces involved in Eq.~\ref{eq:f-ext}. 
\label{gamma_delta_u}}
\end{figure}

$\Delta P^r$ can be realized using the concept of mapping cylinder, see the top right most picture in Figure~\ref{gamma_delta_u}. 
The sequence of inclusion maps associated with each simplex in $\Gamma$
\begin{equation*} 
(B_{i_0}  \cap \cdots \cap B_{i_n}) \hookrightarrow (B_{i_0} \cap \cdots \cap B_{i_{n-1}}) \hookrightarrow \cdots \hookrightarrow (B_{i_0} \cap \cdots \cap B_{i_{n-k}}), 
\label{sequence-inclusion}
\end{equation*}
induces an iterated mapping cylinder. $\Delta P^r $ is obtained by gluing these iterated mapping cylinders over all simplices in $\Gamma$
each simplex $\Delta^k$ having $p_i$ as a vertex in the  barycentric subdivision of $\sigma$ is mapped into $B(p_i, r)$ under $f$
along their boundaries, see~\cite{Hatcher} for details.
There is a canonical projection $\Delta p: \Delta P^r \rightarrow \Gamma$ induced by projecting each finite intersection to 
its corresponding vertex in $\Gamma$. Consider the mapping cylinder $M_{\Delta p}$. The Nerve Lemma is proved in~\cite{Hatcher} 
by showing $M_{\Delta p}$ deformation retracts to $\Delta P^r$. In fact, the deformation retraction described 
there maps a simplex $\Delta^k \in \Gamma$ to the part of $\Delta P^r$ defined over the same $\Delta^k$, 
namely $\Delta q = e_{ \Delta P^r} \circ i_{\Gamma}$ is a homotopy equivalence and maps a simplex $\Delta^k \in \Gamma$
into the iterated mapping cylinder defined by the sequence of inclusion map associated with $\Delta^k$. 

On the other hand, $\Delta P^r$ can also be considered as the quotient space of the disjoint union of all the products 
$B_{i_0}  \cap \cdots \cap B_{i_n} \times \Delta^n $, as the subscripts range
over set of $n+1$ distinct indices and any $n\geq 0$, with the identifications over the faces 
of $\Delta^n$ using inclusions $B_{i_0} \cap \cdots \cap B_{i_n} \hookrightarrow B_{i_0} \cap \cdots 
\cap  \hat{B}_{i_j} \cap \cdots \cap B_{i_n}$ where $\hat{~}$ means the corresponding term is missing.
From this viewpoint, any point $x\in P^r$ has a fiber $\pi^{-1}(x)$ in $\Delta P^r$ defined as follows. 
$\pi^{-1}(x) = \{\sum_i t_i x_i\}$ where  $\sum_i t_i = 1$ and $t^i \geq 0$, and $x_i$ is a copy of $x$ in $B_i$
for those $B_i$ containing $x$. see the bottom left most picture in Figure~\ref{gamma_delta_u}. It is easy to see
that $P^r $ can be embedded into $\Delta P^r$ as a section of $\Delta P^r$, in particular  
$\pi$ is a homotopy equivalence. Thus $f$ is a homotopy equivalence. 

Observe that each point $y$ in an iterated mapping cylinder over some simplex
$\Delta^k=( B_{i_0}  \cap \cdots \cap B_{i_n}, \cdots, B_{i_0} \cap \cdots \cap
B_{i_{n-k}})$ in $\Gamma$ is in the fiber $\pi^{-1}(x)$ for some $x$ in  $B_{i_0}$. 
In other words, if $\Delta^k$ is in the closure of the star of a point $p\in P$ in $\Gamma$,  
then any point $y$ in the iterated mapping cylinder over $\Delta^k$ is in the fiber of 
a point $x \in B(p, r)$. Now consider a simplex $\sigma \in \Cech^{2r}(P)$. 
Any simplex
in its barycentric subdivision much be in the closure of the star of some vertex of $\sigma$. 
Thus $\sigma$, under the map $\Delta_q \circ h$, is mapped into the union of the iterated mapping cylinders
defined over the simplices in the barycentric subdivision of $\sigma$, and its image, under the map $\pi$, 
is further mapped into $\cup_{p\in \Vert(\sigma)}B(p,r)$. 

In addition, it is clear that the map $f$ can fix each vertex in 
$\Cech^{2r}(P)$. This proves the proposition. 
\end{proof}

\end{document}